\documentclass{amsart}
\usepackage{amsfonts}
\usepackage{amsmath}
\usepackage{amssymb}
\usepackage{graphicx}%
\usepackage[english]{babel}
 \usepackage[latin1]{inputenc}

\theoremstyle{plain}

\newtheorem{definition}{Definition}

\newtheorem{lemma}{Lemma}

\newtheorem{proposition}{Proposition}

\newtheorem{theorem}{Theorem}
\numberwithin{equation}{section}

\renewcommand{\t}{\noindent}

\newcommand{\x}{\langle x\rangle^{\alpha}}
  \newcommand{\maxi}[1]{\left\langle #1\right\rangle }
 \let\maxi=\maxi

\DeclareMathOperator{\s}{\mathbf{S}}
\DeclareMathOperator{\f}{\mathcal{F}}
\newcommand{\Qp}{\mathbb{Q}_p}  

\newcommand{\C}{\mathbb{C}}  
\newcommand{\R}{\mathbb{R}}  
\newcommand{\Q}{\mathbb{Q}}  
\newcommand{\Z}{\mathbb{Z}}  

\providecommand{\abs}[1]{\lvert#1\rvert}

\begin{document}
\title[$p$-adic Markov process and the problem of the first return over balls]{ $p$-adic Markov process and the problem of the first return over balls}
\author{O. F. Casas-S\'{a}nchez}
\address{Universidad Pedag\'ogica y Tecnol\'ogica de Colombia, Escuela de Matem\'{a}ticas y Estadística\\
Tunja, Colombia}
\email{oscar.casas01@uptc.edu.co}
\author{J. Galeano-Pe\~{n}aloza}
\address{Universidad Nacional de Colombia, Departamento de Matem\'aticas\\
Ciudad Universitaria, Bogot\'a D.C., Colombia}
\email{jgaleanop@unal.edu.co}
\author{J. J. Rodr\'{\i}guez-Vega}
\address{Universidad Nacional de Colombia, Departamento de Matem\'aticas\\
	Ciudad Universitaria, Bogot\'a D.C., Colombia}
\email{jjrodriguezv@unal.edu.co}
\begin{abstract}
Let $\x=(\max\{|x|_{p},p^r\})^{\alpha}$ and $H^{\alpha}\varphi=\f^{-1}[(\maxi \xi^\alpha -p^{r\alpha})\f\varphi]$, in this article we study the Markov process associated to this operator and the first passage time problem associated to $H^{\alpha}$.
\end{abstract}
\maketitle

\section{Introduction}

Avetisov et al. have constructed a wide variety of models of ultrametric diffusion constrained
by hierarchical energy landscapes (see \cite{Avetisov1}, \cite{Avetisov2}). From a mathematical point of
view, in these models the time-evolution of a complex system is described by a p-adic
master equation (a parabolic-type pseudodifferential equation) which controls the time  evolution
of a transition function of a random walk on an ultrametric space, and the
random walk describes the dynamics of the system in the space of configurational states
which is approximated by an ultrametric space ($\Qp$).

The problem of the first return in dimension 1 was studied in \cite{Av-2}, and in arbitrary
dimension in \cite{CHA} and \cite{TZ}. In these articles, pseudodifferential operators with radial symbols were
considered. More recently, Chac\'{o}n-Cort\'{e}s \cite{CHA2} considers pseudodifferential operators over $\Qp^4$ with non-radial symbol; he studies the problem of first return for a random walk $X(t, w)$ whose density distribution satisfies certain diffusion equation.

In this paper we define the operator 
$$ H^{\alpha}\varphi=\f^{-1}[(\maxi \xi^\alpha -p^{r\alpha})\f\varphi],$$
for $ \varphi\in \s(\Qp)$, where $\maxi{\xi}=\max\{|\xi|_p,p^r\}$.  We also define the heat-kernel $Z_r$ as
		\begin{equation}
 		Z_r(x,t):=\int_{\Qp} \chi(-x\xi)\,e^{-t\bigl(\maxi \xi^\alpha-p^{r\alpha}\bigr)}\,d\xi,
 		\end{equation}
heat kernels of this type have been studied in \cite{CR}, we show that function
\begin{align}\label{solucionclas}
u(x,t)&=Z_r(x,t)\ast\Omega(\left| x\right| _p)
		=\int\limits_{\Qp}\chi(-x\xi)e^{-t\left( \maxi{\xi}^\alpha-p^{r\alpha}\right)\noindent }\Omega(|\xi|_p)\ d\xi
\end{align}
is a solution of Cauchy problem
	\begin{equation}
	\begin{cases}
	u\in C\bigl([0,\infty],\s(\Qp)\bigr)\cap C^1\bigl([0,\infty],L^2(\Qp)\bigr)\\
	\dfrac{\partial u}{\partial t}(x,t)+\bigl(H^{\alpha}\,u\bigr)(x,t)=0,\quad
	x\in \Qp,\quad t\in (0,T], \quad \alpha >0 \\
	u(x,0)=\Omega(|\xi|_p), 
	\end{cases}
	\end{equation}
and we show that $Z_r(x,t)$ is the transition density of a time and space homegeneous Markov process, which is bounded, right-continuous and has no discontinuities other than jumps.

Finally, we study the first passage time problem associated to the operator $H^{\alpha}$.

\section{\label{Section1}Preliminaries}

In this section we fix the notation and collect some basic results on $p$-adic
analysis that we will use through the article. For a detailed exposition on
$p$-adic analysis the reader may consult \cite{Albeverio2010, Taibleson, Vladimirov}.

\subsection{The field of $p$-adic numbers}

Along this article $p$ will denote a prime number. The field of $p-$adic numbers $\Qp$ is defined as the completion of the field of rational numbers $\Q$ with respect to the $p-$adic norm
$|\cdot|_{p}$, which is defined as
\[
|x|_{p}=
\begin{cases}
0 & \text{if }x=0,\\
p^{-\gamma} & \text{if }x=p^{r}\frac{a}{b},
\end{cases}
\]
where $a$ and $b$ are integers coprime with $p$. The integer $\gamma:=ord(x)$,
with $ord(0):=+\infty$, is called the $p-$\textit{adic order of} $x$.

 Any $p-$adic number $x\neq0$ has a unique expansion $x=p^{ord(x)}\sum_{j=0}^{\infty}x_{j}p^{j}$, where $x_{j}\in\{0,1,2,\dots
,p-1\}$ and $x_{0}\neq0$. By using this expansion, we define \textit{the
fractional part of }$x\in\Qp$, denoted $\{x\}_{p}$, as the rational
number
\[
\{x\}_{p}=%
\begin{cases}
0 & \text{if }x=0\text{ or }ord(x)\geq0,\\
p^{\text{ord}(x)}\sum_{j=0}^{-ord(x)-1}x_{j}p^{j} & \text{if }ord(x)<0.
\end{cases}
\]
For $r\in\mathbb{Z}$, denote by $B_{r}(a)=\{x\in\Qp:|x-a|_{p}\leq p^{r}\}$ \textit{the ball of radius }$p^{r
}$ \textit{with center at} $a\in\Qp$, and
take $B_{r}(0):=B_{r}$. 

\subsection{The Bruhat-Schwartz space}
A complex-valued function $\varphi$ defined on $\Qp$ is
\textit{called locally constant} if for any $x\in\Qp$ there
exists an integer $l(x)\in\mathbb{Z}$ such that
\begin{equation}
\varphi(x+x^{\prime})=\varphi(x)\text{ for }x^{\prime}\in B_{l(x)}.
\label{local_constancy}%
\end{equation}
The space of locally constant functions is denoted by $\mathcal{E}(\Qp)$.  A function $\varphi:\Qp\rightarrow\C$ is called a
\textit{Bruhat-Schwartz function (or a test function)} if it is locally
constant with compact support. The $\C$-vector space of
Bruhat-Schwartz functions is denoted by $\s(\Qp)$. For
$\varphi\in\s(\Qp)$, the largest of such number
$l=l(\varphi)$ satisfying (\ref{local_constancy}) is called \textit{the
exponent of local constancy of} $\varphi$.

Let $\s^{\prime}(\Qp)$ denote the set of all
functionals (distributions) on $\s(\Qp)$. All
functionals on $\s(\Qp)$ are continuous.

Set $\chi(y)=\exp(2\pi i\{y\}_{p})$ for $y\in\Qp$. The map
$\chi(\cdot)$ is an additive character on $\Qp$, i.e. a continuos
map from $\Qp$ into $S$ (the unit circle) satisfying $\chi
(y_{0}+y_{1})=\chi(y_{0})\chi(y_{1})$, $y_{0},y_{1}\in\Qp$.

\subsection{Fourier transform}

Given $\xi$ and $x\in
\Qp$, the
Fourier transform of $\varphi\in\s(\Qp)$ is defined as
\[
(\f\varphi)(\xi)=\int_{\Qp}\chi(\xi 
x)\varphi(x)dx\quad\text{for }\xi\in\Qp,
\]
where $dx$ is the Haar measure on $\Qp$ normalized by the
condition $vol(B_{0})=1$. The Fourier transform is a linear isomorphism
from $\s(\Qp)$ onto itself satisfying $(\f(\f\varphi))(\xi)=\varphi(-\xi)$. We will also use the notation
$\f_{x\rightarrow\xi}\varphi$ and $\widehat{\varphi}$\ for the
Fourier transform of $\varphi$.

The Fourier transform $\f\left[  f\right]  $ of a distribution
$f\in\s^{\prime}\left(  \Qp\right)  $ is defined by%
\[
\left(  \f\left[  f\right]  ,\varphi\right)  =\left(  f,\f%
\left[  \varphi\right]  \right)  \text{ for all }\varphi\in\s\left(
\Qp\right)  \text{.}%
\]
The Fourier transform $f\rightarrow\f\left[  f\right]  $ is a linear
isomorphism from $\s^{\prime}\left(  \Qp\right)
$\ onto $\s^{\prime}\left(  \Qp\right)  $. Furthermore,
$f=\f\left[  \f\left[  f\right]  \left(  -\xi\right)
\right]  $.

\section{Pseudodifferential operators}

\begin{definition}
	For all $\alpha\in \C$ we define the following pseudodifferential operator
	\begin{equation}
	H^{\alpha}\varphi=\f^{-1}[(\maxi \xi^\alpha -p^{r\alpha})\f\varphi], \quad \varphi\in \s(\Qp).
	\end{equation}
	Where $\displaystyle{\maxi \xi^\alpha=\max\{|\xi|_p, p^{r}\}}$.
\end{definition}
\t It is clear that the map $H^{\alpha}:\s(\Qp)\rightarrow \s(\Qp)$ is continuous.  Also it is possible to show that the pseudodifferential operator $H^{\alpha}$ has the following representation integral
	\begin{equation}\label{repint}
	(H^{\alpha}\varphi)(x)=\dfrac{1-p^\alpha}{1-p^{\alpha+1}}\left[ p^{r(\alpha+1)}\int\limits_{|y|_p\leq p^{-r}}\varphi(x-y)-\varphi(x)\ dy-p^{(\alpha+1)}\int\limits_{|y|_p\leq p^{-r}}\dfrac{\varphi(x-y)-\varphi(x)}{|y|_p^{\alpha+1}}\ dy\right]. 
	\end{equation}

\begin{definition} Set $\alpha_{k}:=\dfrac{2k\pi i}{\text{ln}\ p}$, $k\in \Z$,
\begin{equation*}
K_{\alpha}(x):=\begin{cases}
\left[\dfrac{1-p^{\alpha}}{1-p^{-\alpha-1}}\abs{x}_{p}^{-\alpha-1}+p^{r(\alpha+1)}\dfrac{1-p^{\alpha}}{1-p^{\alpha+1}}\right]\Omega(p^{r} \abs{x}_{p}), & \text{ for }\alpha \not =-1+\alpha_{k}\\
(1-p^{-1})\Omega(p^{r}\abs{x}_{p})((1-r)- \log_{p} \abs{x}_{p}) & \text{ for } \alpha=-1+\alpha_{k},
\end{cases}
\end{equation*}
and for $\alpha=0$ we define $K_0=\delta.$
\end{definition}

After some calculations it is possible to show the following result.
\begin{theorem}\label{transf-K}
		The Fourier transform of $K_{\alpha}$ is given by $\langle \xi \rangle ^{\alpha}$ for all $\alpha \in \mathbb{C}$.
\end{theorem}

\begin{definition}\label{radial}
For $x\in \Qp, t\in \R$	the heat kernel is defined as
	\begin{equation}
	Z_r(x,t):=\int_{\Qp} \chi(-x\xi)\,e^{-t\bigl(\maxi \xi^\alpha-p^{r\alpha}\bigr)}\,d\xi.
	\end{equation}
The following properties are  proved in \cite{CR}.
\end{definition}
\begin{lemma}\label{prop of Z}
	For $\alpha>0$, $t>0$, the following assertions hold.
	\begin{enumerate}
		\item 
		$Z_r(x,t)\in C(\Qp,\R)\cap L^1(\Qp)\cap L^2(\Qp)$, for $t>0$.
		\item 
		$Z_r(x,t)\geq0$ for all $x\in\Qp$.
		\item 
		$\displaystyle {\int_{\Qp} Z_r(x,t)\,dx=\int_{|x|_p\leq p^{-r}} Z_r(x,t)\,dx=1}$.
		\item 
		$\displaystyle\lim_{t\to 0^{+}} Z_r(x,t)\ast\varphi(x) = \varphi(x)$, for $\varphi\in \s(\Qp)$.
		\item 
		$Z_r(x,t)\ast Z(x,t')=Z(x,t+t')$, for $t,t'>0$.
		\item 
		$Z_r(x,t)\leq Ct|x|_p^{-1}\left( \maxi{px^{-1}}^{\alpha}-p^{r\alpha}\right)$.
	\end{enumerate}
\end{lemma}

\t If 	we set for $\varphi\in \s(\Qp)$
	\begin{equation}\label{def u}
	u(x,t):=
	\begin{cases}
	Z_r(x,t)\ast\varphi(x), & \text{if } t>0\\
	\varphi(x), & \text{if } t=0,
	\end{cases}
	\end{equation}
then it is easy to see that	$u(x,t)\in \s(\Qp)$ for $t\geq 0$, and also it is possible to show that for $t\geq 0$, $\alpha>0$
	\[
	H^{\alpha}(u(x,t))=\f_{\xi\rightarrow x}^{-1} \left[ 
	(\maxi{\xi}^\alpha-p^{r\alpha})e^{-t(\maxi{\xi}^\alpha-p^{r\alpha})}\widehat{\varphi}(\xi) \right]. 
	\]

\begin{theorem}\label{solution}
	Consider the following Cauchy problem
	\begin{equation}
	\begin{cases}
	u\in C\bigl([0,\infty],\s(\Qp)\bigr)\cap C^1\bigl([0,\infty],L^2(\Qp)\bigr)\\
	\dfrac{\partial u}{\partial t}(x,t)+\bigl(H^{\alpha}\,u\bigr)(x,t)=0,\quad
	x\in \Qp,\quad t\in (0,T], \quad \alpha >0 \\
	u(x,0)=\varphi (x), \quad \varphi\in \s(\Qp),
	\end{cases}
	\end{equation}
	then the function $u(x,t)$ defined in (\ref{def u}) is a solution.
\end{theorem}
\begin{proof}
	See Theorem 3.14 in \cite{CR}.
\end{proof}


\section{$p$-adic Markov process over balls}
The space $(\Qp, |\cdot|_p)$ is a complete non-Archimedian metric space.  Let $\mathcal{B}$ the Borel $\sigma$-algebra of $\Qp$; thus $(\Qp, \mathcal{B}, dx)$ is a measure space.  By using the terminology and results of \cite[Chapters 2, 3]{Dyn}, we set
\[
p(t,x,y):=Z_r(x-y,t),\quad t>0,\ \ x,y\in\Qp
\]
and
\[
P(t,x,B)=\begin{cases}
\displaystyle{\int\limits_{B}p(t,x,y)\, dy} & \text{for } t>0,\quad x\in\Qp,\quad B\in\mathcal{B}\\
1_B(x) & \text{for $t=0.$} 
\end{cases}
\]
\begin{lemma}
	With the above notation the following assertions hold:
	\begin{enumerate}
		\item $p(t,x,y)$ is a normal transition density. 
		\item $P(t,x,B)$ is a normal transition function.
	\end{enumerate}
\end{lemma}
\begin{proof} 
	The result follows from Theorem 4.1 (see \cite[Section 2.1]{Dyn}, for further details).
\end{proof}

\begin{lemma}
	The transition function $P(t,x,B)$ satisfies the following two conditions:\\
	(i) \textbf{L(B)} For each $u\geq 0$ and compact $B$,
	\[
	\lim\limits_{|x|_p\to\infty}\sup\limits_{t\leq u}P(t,x,B) = 0.
	\]
	(ii) \textbf{M(B)} For each $\epsilon>0$ and compact $B$,
	\[
	\lim\limits_{t\to 0^+} \sup\limits_{x\in B}P(t,x,\Qp\setminus B_\epsilon(x))=0.
	\]
\end{lemma}
\begin{proof}
	(i) By Lemma \ref{prop of Z} (6), we have
	\begin{align*}
	P(t,x,B)&=\int\limits_{B}Z_r(x-y,t)dy\\
	&\leq Ct\int\limits_{B}|x-y|_p^{-1}\left( \maxi{p(x-y)^{-1}}^{\alpha}-p^{r\alpha}\right)dy\\
	&\qquad\text{ for $x\in\Qp\setminus B$, we have $|x|_p=|x-y|_p$}\\
	&= Ct|x|_p^{-1}\left( \maxi{px^{-1}}^{\alpha}-p^{r\alpha}\right)\int\limits_{B}dy.\\
	\end{align*}
	Therefore, $\lim\limits_{|x|_p\to\infty}\sup\limits_{t\leq u}P(t,x,B) = 0.$\\
	(ii) By using Lemma \ref{prop of Z} (6), $\alpha>0$,  we have
	\begin{align*}
	P(t,x,\Qp\setminus B_\epsilon(x))&\leq Ct\int\limits_{|x-y|> \epsilon}|x-y|_p^{-1}\left( \maxi{p(x-y)^{-1}}^{\alpha}-p^{r\alpha}\right)dy\\
	&=Ct\int\limits_{|z|> \epsilon}|z|_p^{-1}\left( \maxi{p\ z^{-1}}^{\alpha}-p^{r\alpha}\right)dz
	\end{align*}
	if $p^{-r-1}\leq\epsilon<|z|_p$ or $\epsilon< p^{-r-1}\leq|z|_p$, then $\maxi{p\ z^{-1}}^{\alpha}=p^{r\alpha}$ and
	\[
	\int\limits_{|z|> \epsilon}|z|_p^{-1}\left( \maxi{p\ z^{-1}}^{\alpha}-p^{r\alpha}\right)dz=0.
	\]
	Therefore,
	\begin{align*}
	P(t,x,\Qp\setminus B_\epsilon(x))&\leq Ct\int\limits_{|x-y|> \epsilon}|x-y|_p^{-1}\left( \maxi{p(x-y)^{-1}}^{\alpha}-p^{r\alpha}\right)dy\\
	&=Ct\int\limits_{p^{-r-1}>|z|> \epsilon}|z|_p^{-1}\left( |p\ z^{-1}|_p^{\alpha}-p^{r\alpha}\right)dz\\
	&\leq Ctp^{-1}\int\limits_{p^{-r-1}>|z|> \epsilon}|z|_p^{-1-\alpha}dz\\
	&= Ctp^{-1}C_1
	\end{align*}
	Therefore, $\lim\limits_{t\to 0^+} \sup\limits_{x\in B}P(t,x,\Qp\setminus B_\epsilon(x))=0.$
\end{proof}

\begin{theorem}
	$Z_r(x,t)$ is the transition density of a time and space homogeneous Markov process, called $\mathfrak{T}(t,\omega)$, which is bounded, right-continuous and has no discontinuities other than jumps.	
\end{theorem}
\begin{proof}
	The result follows from \cite[Theorem 3.6]{Dyn} by using that $(\Qp,|x|_p)$ is a semi- compact space, i.e., a locally compact Hausdorff space with a countable base, and $P (t, x, B)$ is a normal transition function satisfying conditions{ \it\textbf{L(B)}} and {\it \textbf{M(B)}}.
\end{proof}


\section{The first passage time}

\t By Proposition \ref{solution}, the function
\begin{align}\label{solucionclas}
u(x,t)&=Z_r(x,t)\ast\Omega(\left| x\right| _p)=\int\limits_{\Qp}\chi(-x\xi)e^{-t\left( \maxi{\xi}^\alpha-p^{r\alpha}\right) }\Omega(|\xi|_p)\ d\xi
\end{align}
is a solution of
\begin{equation}
\begin{cases}
\dfrac{\partial u}{\partial t}(x,t)+(H^{\alpha}u)(x,t)=0, & x\in\Qp,\ t>0,\\
u(x,0)=\Omega(\left| x\right| _p). &
\end{cases}
\end{equation}

\t Among other properties, the function $u(x, t) = Z_r(x,t)\ast\Omega(\left| x\right|_p),\ t\geq 0$, is pointwise differentiable in $t$ and, by using the Dominated Convergence Theorem, we can show that  its derivative is given by the formula
	\begin{equation}\label{derivada}
	\dfrac{\partial u}{\partial t}(x, t) = \int\limits_{\Qp}\chi_p(-x \xi)\left( \maxi{\xi}^\alpha-p^{r\alpha}\right) e^{-t\left( \maxi{\xi}^\alpha-p^{r\alpha}\right) }\Omega(|\xi|_p)\ d\xi.
	\end{equation}

\begin{lemma}
	If $\alpha>0$ and $r<0$, then
	\[
	0<-\int\limits_{1<|y|_p\leq p^{-r}}K_\alpha(y)dy<1.
	\]
\end{lemma}
\begin{proof}
	\begin{align*}
	-\int\limits_{1<|y|_p\leq p^{-r}}K_\alpha(y)dy&=\dfrac{1-p^\alpha}{1-p^{\alpha+1}}\left[p^{\alpha+1}\int\limits_{1<|y|_p\leq p^{-r}}\dfrac{1}{|y|_p^{\alpha+1}}dy- p^{r(\alpha+1)}\int\limits_{1<|y|_p\leq p^{-r}}dy\right] \\
	&<\dfrac{1-p^\alpha}{1-p^{\alpha+1}}\left[p^{\alpha+1}\int\limits_{1<|y|_p}\dfrac{1}{|y|_p^{\alpha+1}}dy- p^{r(\alpha+1)}(p^{-r}-1)\right]\\
	&=\dfrac{1-p^{-1}}{1-p^{-\alpha-1}}- \dfrac{1-p^\alpha}{1-p^{\alpha+1}}p^{r\alpha}(1-p^{r})\\
	&=1-\dfrac{1-p^{\alpha}}{1-p^{\alpha+1}}\left( 1+p^{r\alpha}(1-p^{r})\right) \\
	&<1.
	\end{align*}
	Now
	\begin{align*}
	-\int\limits_{1<|y|_p\leq p^{-r}}K_\alpha(y)dy
	&=\dfrac{1-p^\alpha}{1-p^{\alpha+1}}\left[p^{\alpha+1}\int\limits_{1<|y|_p\leq p^{-r}}\dfrac{1}{|y|_p^{\alpha+1}}dy- p^{r(\alpha+1)}\int\limits_{1<|y|_p\leq p^{-r}}dy\right] \\
	&>\dfrac{p^\alpha(p-1)(1-p^{r\alpha})}{p^{\alpha+1}-1}+\dfrac{p^{r\alpha}(1-p^\alpha)}{p^{\alpha+1}-1}\\
	&>0.
	\end{align*}
\end{proof}

\t The rest of this section is dedicated to the study of the following random variable.

\begin{definition}
	\label{randonvarible} The random variable $\tau_{\Omega(|x|_p)}(\omega) :\mathfrak{Y}\to \R_+$ defined by \\
{\small $$\inf\{t> 0\mid \mathfrak{T}(t,\omega) \in\Omega(|x|_p) \text{ and there exists $t'$ such that  $0 <t'< t$ and $\mathfrak{T}(t',\omega) \notin\Omega(|x|_p)$}\}$$}
\t is called the first passage time of a path of the random process $\mathfrak{T}(t,\omega)$ entering the domain $\Omega(|x|_p)$.
\end{definition}

\begin{lemma}\label{lemma12}
	The probability density function for a path of $\mathfrak{T}(t,\omega)$ to enter into $\Omega(|x|_p)$ at the instant of time $t$, with the condition that $\mathfrak{T}(0,\omega)\in\Omega(|x|_p)$ is given by
	\begin{equation}\label{gequa}
	g(t)=\int\limits_{1<|y|_p\leq p^{-r}}K_\alpha(y)u(y,t)dy.
	\end{equation}
\end{lemma}
\begin{proof}
	We first note that, for $x,y\in\Omega(|z|_p)$, we have
	\begin{align*}
	u(x-y,t)&=\int\limits_{\Omega(|\xi|_p)}\chi_p(-(x-y)\cdot\xi)e^{-t(\maxi{\xi}^\alpha-p^{r\alpha})}\ d\xi\\
	&=\int\limits_{\Omega(|\xi|_p)}e^{-t(\maxi{\xi}^\alpha-p^{r\alpha})}\ d\xi
	=\int\limits_{\Omega(|\xi|_p)}\chi_p(-x\cdot\xi)e^{-t(\maxi{\xi}^\alpha-p^{r\alpha})}\ d\xi\\
	&=u(x,t).
	\end{align*}
	i.e. $u(x-y,t)-u(x,t)\equiv0$ for $x,y\in\Omega(|z|_p)$.
	
	The survival probability, by definition 
	\[
	S(t):=S_{\Omega(|x|_p)}(t)=\int_{\Omega(|x|_p)}u(x,t)d^nx,
	\]
	is the probability that a path of $\mathfrak{T}(t,\omega)$ remains in $\Omega(|x|_p)$ at the time $t$. Because there are no external or internal sources,
	\begin{align*}
	S^{\prime}(t)  &  =%
	\begin{array}
	[c]{l}%
	\text{Probability that a path of }\mathfrak{T}(t,\omega)\\
	\text{goes back to }\Omega(|x|_p)\text{ at the time }t
	\end{array}
	-%
	\begin{array}
	[c]{l}%
	\text{Probability that a path of }\mathfrak{T}(t,\omega)\\
	\text{exits }\Omega(|x|_p)\text{ at the time }t.
	\end{array}
	\label{S}\\
	&  =g(t)-C\cdot S(t)\text{ with }0<C\leq1.\nonumber
	\end{align*}
	by using the derivative \eqref{derivada}
	\begin{align*}
	S^{\prime}(t)&=\underset{\Omega(|x|_p)}{\int}\frac{\partial u
		(x,t)}{\partial t}dx\\
	&=-\dfrac{1-p^\alpha}{1-p^{\alpha+1}}\left[ p^{r(\alpha+1)}\int\limits_{|x|_p\leq 1}\int_{1<|y|_p\leq p^{-r}}u(x-y,t)-u(x,t)\ dydx\right. \\
	&\qquad\qquad\qquad\ \left. -p^{(\alpha+1)}\int\limits_{|x|_p\leq 1}\int_{1<|y|_p\leq p^{-r}}\dfrac{u(x-y,t)-u(x,t)}{|y|_p^{\alpha+1}}\ dydx\right]\\
	&=-\dfrac{1-p^\alpha}{1-p^{\alpha+1}}\left[ p^{r(\alpha+1)}\int\limits_{|x|_p\leq 1}\int_{1<|y|_p\leq p^{-r}}u(x-y,t)\ dydx\right. \\
	&\qquad\qquad\qquad\ \left. -p^{(\alpha+1)}\int\limits_{|x|_p\leq 1}\int_{1<|y|_p\leq p^{-r}}\dfrac{u(x-y,t)}{|y|_p^{\alpha+1}}\ dydx\right]\\
	&\quad+\dfrac{1-p^\alpha}{1-p^{\alpha+1}}\left[ p^{r(\alpha+1)}\int\limits_{|x|_p\leq 1}\int_{1<|y|_p\leq p^{-r}}u(x,t)\ dydx\right. \\
	&\qquad\qquad\qquad\ \left. -p^{(\alpha+1)}\int\limits_{|x|_p\leq 1}\int_{1<|y|_p\leq p^{-r}}\dfrac{u(x,t)}{|y|_p^{\alpha+1}}\ dydx\right].
	\end{align*}
	Now if $y\in\Omega(p^r|y|_p)\setminus\Omega(|y|_p)$ and $x\in\Omega(|x|_p)$, then		$u(x-y,t)=u(y,t)$, consequently
	\begin{align*}
	S'(t)&=\int\limits_{1<|y|_p\leq p^{-r}}K_\alpha(y)u(y,t)dy+\int\limits_{1<|y|_p\leq p^{-r}}K_\alpha(y)dy\int\limits_{|x|_p\leq 1}u(x,t)\ dx\\
	&=g(t)-CS(t),
	\end{align*}
	where $\displaystyle{C=-\int\limits_{1<|y|_p\leq p^{-r}}K_\alpha(y)dy}$.
\end{proof}
\begin{proposition}
	\label{prop4}The probability density function $f(t)$ of the random variable
	$\tau_{\Omega(|x|_p)}(\omega)$ satisfies the non-homogeneous Volterra equation of second kind%
	\begin{equation}
	g(t)=\int_{0}^{\infty}g(t-\tau)f(\tau)d\tau+f(t). \label{eq:Volterra}%
	\end{equation}
\end{proposition}
\begin{proof}
	The result follows from Lemma \ref{lemma12}\ by using the argument given in
	the proof \ of Theorem 1 in \cite{Av-2}.
\end{proof}
\begin{proposition}
	The Laplace transform $G_r(s)$ of $g(t)$ is given by 
	\[
	G_r(s)=\int\limits_{1<|y|_p\leq p^{-r}}K_\alpha(y)\int\limits_{|\xi|_p\leq1}\dfrac{\chi_p(-\xi\cdot y)}{s+(\maxi{\xi}^\alpha-p^{r\alpha})}\ d\xi dy.
	\]
\end{proposition}
\begin{proof}
	We first note that $e^{-st}K_\alpha(y)e^{-t(\maxi{\xi}^\alpha-p^{r\alpha})}\Omega(|\xi|_p)\in \mathcal{L}^1((0,\infty) \times\Omega(p^r|\xi|_p)\setminus \Omega(|\xi|_p)\times\Qp, dtdyd\xi)$ for $s\in\C$ with $Re(s) >0$. The announced formula follows now from \eqref{gequa} and \eqref{solucionclas} by using Fubini's Theorem.
\end{proof}
\begin{definition}
	We say that $\mathfrak{T}(t,\omega)$ is recurrent with respect to $\Omega(|x|_p)$ if 
	\begin{equation}\label{retorno}
	P(\{\omega \in \mathfrak{Y} : \tau_{\Omega(|x|_p)} (\omega)<\infty\}) = 1.
	\end{equation}
	Otherwise, we say that $\mathfrak{T}(t,\omega)$ is transient with respect to $\Omega(|x|_p)$ .
\end{definition}
The meaning of \eqref{retorno} is that every path of $\mathfrak{T}(t,\omega)$ is sure to return to $\Omega(|x|_p)$. If \eqref{retorno} does not hold, then there exist paths of $\mathfrak{T}(t,\omega)$ that abandon $\Omega(|x|_p)$ and never go back.
\begin{theorem}
	For all $\alpha>0$ the processes $\mathfrak{T}(t,\omega)$ is recurrent with respect to $\Omega(|x|_p)$.
\end{theorem}
\begin{proof}
	By Proposition \ref{prop4}, the Laplace transform $\ F(s)$ of $\ f(t)$
	equals $\dfrac{G_r(s)}{1+G_r(s)}$, where $G_r(s)$ is the Laplace transform of $g(t)$,
	and thus 
	\[F\left(  0\right)  =\int_{0}^{\infty}f\left(  t\right)
	dt=1-\frac{1}{1+G_r(0)}.\] 
	Hence in order to prove that $\mathfrak{T}(t,\omega)$ is recurrent is sufficient to show that
	\[
	G_r(0)=\lim\limits_{s\rightarrow0}G_r(s)=\infty,
	\]
	and to prove that it is transient that
	\[
	G_r(0)=\lim\limits_{s\rightarrow0}G_r(s)<\infty.
	\]
	\begin{align*}
	G_r(s)&=\int\limits_{1<|y|_p\leq p^{-r}}\int\limits_{|\xi|_p\leq p^r}\dfrac{K_\alpha(y)\chi(-\xi y)}{s}d\xi dy+\int\limits_{1<|y|_p\leq p^{-r}}\int\limits_{p^r<|\xi|_p\leq 1}\dfrac{K_\alpha(y)\chi(-\xi y)}{s+|\xi|_p^\alpha-p^{r\alpha}}d\xi dy\\
	&=\dfrac{p^r}{s}\int\limits_{1<|y|_p\leq p^{-r}}K_\alpha(y)dy+\int\limits_{1<|y|_p\leq p^{-r}}\int\limits_{p^r<|\xi|_p\leq 1}\dfrac{K_\alpha(y)\chi(-\xi y)}{s+|\xi|_p^\alpha-p^{r\alpha}}d\xi dy\\
	&=\dfrac{p^r}{s}\int\limits_{1<|y|_p\leq p^{-r}}K_\alpha(y)dy+\sum\limits_{k=1}^{-r}\sum\limits_{m=0}^{k-1}\dfrac{p^{k-m}}{s+p^{-m\alpha}-p^{r\alpha}}\int\limits_{|u|_p=1}K_{\alpha}(p^{-k}u)du\\
	&\quad+\sum\limits_{k=1}^{-r}\sum\limits_{m=k}^{-r-1}\dfrac{p^{k-m}}{s+p^{-m\alpha}-p^{r\alpha}}\int\limits_{|u|_p=1}\int\limits_{|v|_p=1}K_{\alpha}(p^{-k}u)\chi(-p^{m-k}uv)dvdu\\
	\end{align*}
	therefore $\displaystyle{\lim\limits_{s\rightarrow0}G_r(s)=\infty}$ and the process $\mathfrak{T}(t,\omega)$ is recurrent.
\end{proof}

\end{document}